\documentclass[letterpaper, draftclsnofoot, onecolumn]{ieeeconf}

\usepackage{soul}
\usepackage[normalem]{ulem}

\usepackage[hidelinks]{hyperref}
\usepackage{bookmark}
\usepackage[latin1]{inputenc}
\usepackage[english]{babel}

\usepackage[dvipsnames]{xcolor}
\usepackage{amsfonts,amssymb,amsmath,color}
\usepackage[draft]{graphicx}
\usepackage[noadjust]{cite}
\usepackage{placeins}
\usepackage{graphicx}
\usepackage{algorithm}
\usepackage[]{algpseudocode}

\usepackage{tikz}
\usetikzlibrary{shapes,arrows}

\newcommand{\map}[3]{#1: #2 \rightarrow #3}

\newcommand{\diag}{\operatorname{diag}}

\newcommand{\EE}{\mathcal{E}} 
\newcommand{\GG}{\mathcal{G}}

\newcommand{\nbrs}{\mathcal{N}}

\newcommand{\real}{{\mathbb{R}}}

\newcommand{\until}[1]{\{1,\ldots,#1\}}

\newcommand{\toN}[1][N]{1,\dots,#1}

\newcommand\oprocendsymbol{\hbox{$\triangle$}}
\newcommand\oprocend{\relax\ifmmode\else\unskip\hfill\fi\oprocendsymbol}

\renewcommand{\theenumi}{(\roman{enumi})}

\newtheorem{theorem}{Theorem}[section]
\newtheorem{proposition}[theorem]{Proposition}

\newtheorem{definition}[theorem]{Definition} \newtheorem{lemma}[theorem]{Lemma}
\newtheorem{remark}[theorem]{Remark} 
 
\newtheorem{problem}[theorem]{Problem}

\def\QEDclosed{\mbox{\rule[0pt]{1.3ex}{1.3ex}}} %

\def\QED{\QEDclosed} %

\newcommand{\bA}{\mathbf{A}}
\newcommand{\btA}{\mathbf{\tA}}

\newcommand{\bI}{\mathbf{I}}
\newcommand{\0}{\mathbf{0}}
\newcommand{\1}{\mathbf{1}}
\newcommand{\ts}{\theta^\star}
\newcommand{\xeq}{x_{\text{eq}}}
\newcommand{\zeq}{z_{\text{eq}}}

\newcommand{\kron}{\otimes}

\makeatletter
\newcommand{\StatexIndent}[1][3]{%
  \setlength\@tempdima{\algorithmicindent}%
  \Statex\hskip\dimexpr#1\@tempdima\relax}
\makeatother

\renewcommand{\lim}{\operatornamewithlimits{lim\vphantom{p}}}

\graphicspath{{figs/}}

\newcommand{\cO}{\mathcal{O}}

\newcommand{\x}{\times}

\newcommand{\R}{\mathbb{R}}
\newcommand{\N}{\mathbb{N}}

\newcommand{\cV}{\mathcal{V}}

\newcommand{\ub}{\mathbf{u}}

\newcommand{\dx}{{d}}
\newcommand{\dz}{{n_z}}

\renewcommand{\st}{\;\mid\;}

\newcommand{\xb}{\mathbf{x}}

\newcommand{\nb}{\mathbf{n}}

\newcommand{\tA}{\tilde{A}}
\newcommand{\ta}{\tilde{a}}

\DeclareMathOperator{\img}{Im}
\DeclareMathOperator{\col}{col}

\author{Michelangelo Bin, Ivano Notarnicola, Lorenzo Marconi, Giuseppe Notarstefano
\thanks{
  The authors are with the Department of Electrical, Electronic and Information Engineering, 
  University of Bologna, Bologna, Italy, \texttt{name.lastname@unibo.it}.
  }
\thanks{
  The research leading to these results has received funding from the European Research Council (ERC) under the European Union's Horizon 2020 research and innovation programme (grant agreement No 638992 - OPT4SMART) 
and from the European Project AirBorne (ICT 780960).}
}

\title{\bf A System Theoretical Perspective to Gradient-Tracking
  Algorithms\\ for Distributed Quadratic Optimization} %

\def \algname/{Gradient Tracking}
\def \conditionedinitializationset/{admissible initialization set}

\begin{document}
\maketitle

\noindent
{\bf \textcopyright 2019 IEEE. Personal use of this material is permitted. Permission from IEEE must be obtained for all other uses, in any current or future media, including reprinting/republishing this material for advertising or promotional purposes, creating new collective works, for resale or redistribution to servers or lists, 
or reuse of any copyrighted component of this work in other works.\\[1em]}

\begin{abstract}
  In this paper we consider a recently developed distributed optimization
  algorithm based on gradient tracking. We propose a system theory framework
  to analyze its structural properties on a preliminary, quadratic optimization
  set-up.  Specifically, we focus on a scenario in which agents in a static
  network want to cooperatively minimize the sum of quadratic cost functions.
  We show that the gradient tracking distributed algorithm for the investigated
  program can be viewed as a sparse closed-loop linear system in which
  the dynamic state-feedback controller includes consensus matrices and optimization
  (stepsize) parameters. The closed-loop system turns out to be not completely
  reachable and asymptotic stability can be shown restricted to a proper
  invariant set. 
  Convergence to the global minimum, in turn, can be obtained only by 
  means of a proper initialization. The proposed system interpretation 
  of the distributed algorithm provides also additional insights on other 
  structural properties and possible design choices that are discussed in 
  the last part of the paper as a starting point for future developments. 
\end{abstract}

\section{Introduction}
\label{sec:intro}
Many optimization algorithms are iterative procedures that can be, thus, framed
as discrete-time dynamical systems.
Usual approaches to prove the convergence of these schemes, even though often
based on descent, Lyapunov-like arguments, do not explicitly and deeply explore
this system theoretical perspective.
The great potential of system theory becomes more evident when noticing that
several algorithms encode a feedback structure in their update laws.
In this paper we propose a system theoretical interpretation of a
state-of-the-art \emph{distributed} optimization algorithm often known as
gradient tracking, see, e.g.,~\cite{dilorenzo2016next,varagnolo2016newton,nedic2017achieving, 
qu2018harnessing,xu2018convergence,xi2018addopt,xin2018linear,scutari2019distributed}.

In this framework agents (systems) in a network cooperate to
minimize the sum of local functions that depend on a common decision
variable. Agents exchange information with neighbors in a given (sparse)
communication graph and cannot rely on any centralized coordinating unit.
We consider a simplified set-up in which the optimization problem is quadratic
and the communication occurs according to a fixed and undirected graph.

Distributed optimization has received a large interest from the control community in the
last decades. Early references on this topic are~\cite{nedic2009distributed,nedic2010constrained} 
where the (sub)gradient method has been successfully combined with consensus averaging
to design a distributed method.
Recently, this approach has been enhanced by introducing a tracking technique 
based on the dynamic average consensus, originally proposed 
in~\cite{zhu2010discrete,kia2018tutorial}. The tracking mechanism
allows agents to obtain a local estimate of the gradient of the entire sum 
of functions, which is then used as a descent direction in the consensus-based
update of the local solution estimate, see, 
e.g.,~\cite{dilorenzo2016next,varagnolo2016newton,nedic2017achieving, 
qu2018harnessing,xu2018convergence,xi2018addopt,xin2018linear,scutari2019distributed}.

First approaches providing a system theoretical perspective to distributed
optimization algorithms are~\cite{wang2010control,wang2011control}.
A framework based on integral quadratic constraints from robust control theory
is proposed in~\cite{lessard2016analysis} to analyze and design (centralized)
iterative optimization algorithms.
In~\cite{hu2017control} authors propose a loop-shaping interpretations for
several existing optimization methods based on basic control elements such as
PID and lag compensators.
The convergence of distributed optimization algorithms by means of proper
semidefinite programs is, instead, discussed in~\cite{sundararajan2017robust}.
A passivity-based approach is proposed in~\cite{hatanaka2018passivity} to
analyze a distributed algorithm with communication delays.

The contributions of this paper are as follows. 
We approach the design of distributed optimization algorithms as a
control problem, by showing how system theoretical tools can be used to
provide new insights on the existing algorithms, and new perspectives for future
extensions. We develop the discussion for a simplified quadratic, unconstrained
optimization problem, that allows us to rely on powerful tools from linear 
regulation theory.
In particular, we cast the optimization algorithm design as a linear control problem 
aiming to steer the state trajectories toward the optimal solution, and we provide 
necessary and sufficient conditions for its solvability.
We show that a class of gradient tracking distributed algorithms fits in the proposed 
framework, which, in turn, provides new insights in terms of structural properties of 
the controlled system.
Specifically, the resulting algorithm, seen as a sparse dynamical system, turns out to be 
not completely reachable and this reflects on the need of a proper initialization and of
the necessity of a stabilizing action in the ``closed-loop'' dynamics. 
The proposed system theoretical perspective suggests that robustness arguments,
customary in control theory, can be used to extends these features also to optimization 
algorithms.%

The paper is organized as follows. In Section~\ref{sec:set-up} we introduce the
distributed optimization set-up and recall the gradient tracking algorithm. In
Section~\ref{sec:control_approach} we describe the system theoretical framework to
solve a quadratic distributed optimization problem which is used in
Section~\ref{sec:GT_revisited} to analyze the gradient tracking algorithm.

\paragraph*{Notation}
We deal with discrete-time dynamical systems of the form $x(t+1) = \phi
(x(t))$. For the sake of readability we omit the time dependency whenever it is
clear from the context and we write $x^+$ in place of $x(t+1)$.
Given a square matrix $M\in\real^{\dx\x\dx}$, we denote by $\sigma(M)$ its
spectrum.
A square matrix is said to be Schur if all its eigenvalues 
lie inside the open unitary disc.
Given a square matrix $F\in\real^{\dx\x\dx}$, a set $\cV \subset \R^\dx$ is 
said to be $F$-invariant if for all $v\in\cV$ it holds $Fv \in \cV$.
We denote by $I_\dx$ the $d\x d$ identity matrix and by $0_\dx$ the $d\x d$
matrix of zeros.  The column vector of $\dx$ ones is denoted by
$\1_\dx$. Moreover, we define $\bI = \1_N \kron I_d$ where $\kron$ is the
Kronecker product.  We omit the dimension of these objects whenever it is clear from
the context. For $x\in\R^{n_1}$ and $z\in\R^{n_2}$, we denote by
$\col(x,z)\in\R^{n_1+n_2}$ their column concatenation. For $S\subset\R^n$ and
$x\in\R^n$, we let $x+S:=\{ z\in\R^n\st z=x+s, \,s\in S \}$.

\section{The Distributed Optimization Framework}
\label{sec:set-up}
In this section we introduce the distributed optimization set-up 
and recall the state-of-the-art gradient tracking algorithm 
that we aim to investigate in this paper.

\subsection{Distributed Optimization Set-up}
We consider the following optimization problem
\begin{align}\label{eq:problem}
  \min_{\theta \in \R^\dx} \: & \: \sum_{i=1}^N f_i ( \theta ),
\end{align}
where, for each $i = \toN$,   $\map{f_i}{\R^\dx}{\R}$ is of the form
\begin{equation}\label{d:fi}
  f_i( \theta) = \dfrac{1}{2} (\theta - \Gamma_i \theta_0 )^\top C_i (\theta - \Gamma_i \theta_0),
\end{equation}
with $C_i\in\R^{\dx \x \dx}$  symmetric and positive-definite, 
$\Gamma_i \in \real^{\dx \times p}$, and where $\theta_0 \in \real^p$ is an offset parameter 
whose role will be  clarified later. In particular, 
problem~\eqref{eq:problem} admits a unique solution given by
\begin{align}\label{d:Sigma}
  \begin{aligned}
  \ts & :=  \Sigma  \theta_0,  & \Sigma&:= \left ( \sum_{i=1}^N C_i \right )^{-1} \sum_{i=1}^N C_i \Gamma_i.
  \end{aligned}
\end{align} 
We focus on iterative procedures to solve~\eqref{eq:problem} that are
\textit{distributed}.  In particular, we assume to have a network of $N$ agents,
each one having access only to partial information about the problem and
exchanging information with a subset of the other agents. Distributed
optimization algorithms are \emph{local} update laws that fulfill the network
constraints and allow agents to eventually converge the optimal solution $\ts$.
Formally, we model the network by means of a connected and undirected graph $\GG = (\until{N}, \EE)$
where $\EE\subseteq \until{N} \times \until{N}$ is the set of edges. If $(i,j) \in \EE$, 
then nodes $i$ and $j$ can exchange information (and, in fact,  $(j,i)\in\EE$). 
We denote by $\nbrs_i := \left\{j \in \until{N} \mid (i,j) \in \EE \right\}$ the set of neighbors 
of node $i$ in $\GG$. We assume that $\nbrs_i$ contains $i$ itself.
As usually done in consensus-based approaches, we consider a matrix
$A \in \real^{N\x N}$ matching the graph $\GG$, i.e., $(i,j)$-th entry
$a_{ij} > 0$ for $(i,j) \in \EE$ while $a_{ij} = 0$ otherwise.
Moreover, $A$ is row stochastic if $A \1_N= \1_N$, while it is column 
stochastic if  $\1_N^\top A = \1_N^\top$.
It can be proved that the spectrum of a row (or column) stochastic matrix lies
in the closed unitary circle and the largest (in norm) eigenvalue is $1$ and is
simple.

In this paper we assume that each agent $i$ maintains a local quantity
$x_i\in\R^\dx$ representing its guess of the optimal solution $\ts$, and it has
only access to gradients of the \emph{local cost function} computed at $x_i$,
i.e., to the quantity
\begin{equation}\label{d:yi}
\begin{split}
  y_i  &= \nabla f_i(x_i) = C_i x_i + Q_i\theta_0 ,
\end{split}
\end{equation}
where $Q_i := -C_i \Gamma_i$.
In these terms, the distributed optimization problem associated to~\eqref{eq:problem} can be cast as follows.
\begin{problem}\label{prob:distributedOpt}
  Find an update law for $x_i$, depending only on the local available information 
  given by the quantities $(x_j,y_j)$ for all $j\in\nbrs_i$, such that
  \begin{equation*}
    \lim_{t\to \infty} x_i(t) = \ts = \Sigma \theta_0,
  \end{equation*}
  for each $i\in\until{N}$.\oprocend
\end{problem}

Problem~\ref{prob:distributedOpt} could be clearly solved in a distributed way
through a consensus algorithm by exploiting equation~\eqref{d:Sigma}. In this
paper we focus on distributed optimization algorithms to solve Problem~\ref{prob:distributedOpt}.
It is worth mentioning that in some applications agents may not know 
$C_i$, $Q_i$ and $\theta_0$ but just the local \emph{measurement} $y_i$.
Notice that, in view of~\eqref{d:yi}, each matrix $C_i$ is directly linked 
to the Lipschitz constant of the corresponding local gradient $\nabla f_i$, 
while the affine terms $Q_i\theta_0$ represent a
partial information on $\theta_0$ that, even if accessible via $y_i$, is not
assumed to be known a priori. On this regard, $\theta_0$ is a parameter
condensing an information which is not known to the agents.

\subsection{The Gradient Tracking Algorithm}
\label{sec:gradient_tracking}
In this subsection, we recall the \emph{gradient tracking algorithm} in its most
basic form. For convenience, we first recall the (centralized) gradient method
applied to a generic instance of~\eqref{eq:problem}.
In the (steepest descent) gradient method, a solution estimate $\theta(t)$ is
iteratively updated according to\footnote{As discussed in the Notation
  paragraph, we omit the time dependence when not strictly necessary.}
\begin{align*}
  \theta^+ = \theta - \gamma \sum_{i=1}^N \nabla f_i (\theta),
\end{align*}
where $\gamma$ is a constant, positive parameter that is usually called \emph{stepsize}.
Convergence results for the class of gradient methods can be found, e.g., 
in~\cite{bertsekas1999nonlinear}.

The gradient tracking distributed algorithm mimics the centralized update by exploiting
a twofold consensus-based mechanism to: \emph{(i)} enforce an agreement among the agents' 
estimates $x_i$ and \emph{(ii) }dynamically track the gradient of the whole cost function
through an auxiliary variable $s_i \in \real^\dx$, called \emph{tracker}.
Formally, it reads
\begin{subequations} \label{eq:gradient_tracking}
\begin{align}
  \label{eq:gradient_tracking_x}
  x_i^+& = \!\! \sum_{j \in\nbrs_i} a_{ij} x_j - \gamma s_i
  \\
  \label{eq:gradient_tracking_s}
  s_i^+ & = \sum_{j \in\nbrs_i} \ta_{ij} s_j + \nabla f_i (x_i^+) - \nabla f_i (x_i),
\end{align}
\end{subequations}
where $a_{ij}$ and $\ta_{ij}$ are entries of a row stochastic 
matrix $A \in\real^{N\x N}$ and of a column stochastic matrix $\tA\in\real^{N\x N}$, 
respectively, while $\gamma >0$ is a (constant) stepsize.

Several versions of the gradient tracking algorithm have been analyzed for
generic, nonlinear, and possibly constrained versions of
problem~\eqref{eq:problem}, see, e.g.,
\cite{dilorenzo2016next,varagnolo2016newton,nedic2017achieving, 
qu2018harnessing,xu2018convergence,xi2018addopt,xin2018linear,scutari2019distributed}. 
For example, in~\cite{xin2018linear}
it is shown that, under strong convexity of the local cost functions $f_i$, and
Lipschitz continuity of their gradients, the sequence
$\{(x_1(t), \ldots, x_N(t))\}_{t\ge0}$ generated by
algorithm~\eqref{eq:gradient_tracking}, with $x_i(0)$ arbitrary,
$s_i(0) = \nabla f_i(x_i(0))$ and for a sufficiently small stepsize $\gamma$,
converges to the optimal solution $\ts$ of~\eqref{eq:problem}.

\begin{remark}\label{rmk:gt_init}
  An interesting property of the states $s_i$ is that, by summing over $i$ the
  update~\eqref{eq:gradient_tracking_s}, we can exploit the column stochasticity
  of the weights $\ta_{ij}$ to obtain
\begin{align} \label{eq:conservation}
  \sum_{i=1}^N s_i^+ - \sum_{i=1}^N \nabla f_i (x_i^+)
  =
  \sum_{i=1}^N s_i - \sum_{i=1}^N \nabla f_i (x_i).
\end{align}
Specifically, condition~\eqref{eq:conservation} holds at $t=0$. Moreover, it
also holds for any \emph{consensual} asymptotic value of $(x_i,s_i)$.  By
assuming $\lim_{t\to\infty} x_i(t) = x_i^{\infty} = x^{\infty}$, for all
$i= \toN$, it can be shown that the asymptotic value of the tracker is
$\lim_{t\to\infty} s_i(t) = s_i^{\infty} = 0$, $\forall \, i$ (recall that
weights $a_{ij}$ in~\eqref{eq:gradient_tracking_x} sum up to one, i.e., $A$ is
row stochastic). Thus, we have
\begin{align*}
  \sum_{i=1}^N s_i^{\infty} -
  \sum_{i=1}^N \nabla f_i ( x_i^{\infty})
  & =
  -
  \sum_{i=1}^N \nabla f_i ( x^{\infty})
  \\
  & =
  \sum_{i=1}^N s_i(0) -
  \sum_{i=1}^N \nabla f_i (x_i(0)).
\end{align*}
This, in turn, shows that if the initialization of each $s_i$ is arbitrary, so
that the last line is not zero, there is no chance that a consensual asymptotic
value is stationary (hence optimal) for problem~\eqref{eq:problem}.\oprocend
\end{remark}

The distributed algorithm described by~\eqref{eq:gradient_tracking} does not
enjoy the usual ``state-space'' structure of dynamical systems, since the
updated $s_i^+$ depends on $x_i^+$. Thus, we consider the change of variable
$z_i := s_i - \nabla f_i (x_i)$, so that algorithm~\eqref{eq:gradient_tracking}
can be equivalently rewritten as
\begin{subequations}
\label{eq:gradient_tracking_xz_distribtued}
\begin{align}
  x_i^+ & = \sum_{j \in\nbrs_i} a_{ij} x_j - \gamma \left ( z_i + \nabla f_i(x_i) \right )
  \\
  z_i^+ & = \sum_{j \in\nbrs_i} \ta_{ij} z_j + \sum_{j \in\nbrs_i} \ta_{ij} \nabla f_j (x_j) - \nabla f_i (x_i).
\end{align}
\end{subequations}
In these new coordinates, the correct initialization becomes
$z_i(0)  0$, for all $i=\toN$.  Also, this reformulation does not alter the
distributed nature of the algorithm.

Let $x := (x_1, \ldots, x_N)$, 
$z := (z_1, \ldots, z_N)$ and $y := (y_1, \ldots, y_N)$ 
and compactly rewrite~\eqref{eq:gradient_tracking_xz_distribtued} as
\begin{equation}
\label{eq:gradient_tracking_xz}
\begin{array}{lcl}
  x^+ & = & \bA x - \gamma z  - \gamma  y
  \\
  z^+ & = & \btA z + (\btA  -I_{dN}) y
  \\
  y & = & \nabla F(x),
\end{array}
\end{equation}
in which $\bA = A \kron I_d$, $\btA = \tA \kron I_d$, $z(0)=0_{d N}$, and $\nabla F(x)$ denotes a column
vector stacking the local gradients, i.e., $\nabla F(x) = (\nabla f_1(x_1), \ldots, \nabla f_N(x_N))$.
Notice that, for the considered quadratic scenario the output map $y$ is (as expected) affine 
in $x$ and this structure will be exploited later. %

\section{A System Theoretical Approach to Quadratic Distributed Optimization}
\label{sec:control_approach}

In this section, we approach the design of a distributed algorithm solving optimization 
problem~\eqref{eq:problem} from a system theoretical perspective. 
Specifically, we approach Problem~\ref{prob:distributedOpt} as a generic \emph{set-point control problem}, 
by giving necessary and sufficient conditions for its solvability. 
For the sake of presentation, in this part we intentionally avoid dealing explicitly 
with network constraints. We will   discuss constructive choices of the different 
degrees of freedom that are consistent with the network constraints in Section~\ref{sec:GT_revisited}.

\subsection{The Underlying Control Problem}
The local estimates $x_i$ can be seen as the \emph{state} of $N$ controlled
plants
\begin{equation}\label{ctr:s:xi}
  x^+_i = u_i, \hspace{0.5cm} i = \toN.
\end{equation}
The control goal consists of finding
a suitable control input $u = (u_1,\dots,u_N)$ such that each
controlled plant asymptotically converges to the optimal solution $\ts$ of
problem~\eqref{eq:problem}.
We further point out that in this regulation setting the target equilibrium $\ts$ is not
available for feedback.

By letting $x:=(x_1,\dots,x_N)$
and $y:=(y_1,\dots,y_N)$, the overall controlled plant, obtained by stacking the local dynamics~\eqref{ctr:s:xi}
and the local measurements~\eqref{d:yi}, reads as
\begin{equation}\label{ctr:s:x}
\begin{array}{lcl} 
x^+ &=& u
\\
y&=& C  x + Q  \theta_0,
\end{array}
\end{equation}
where $C:=\diag(C_1,\dots,C_N)$ and $Q:=\col(Q_1,\dots,Q_N)$, with 
$C_i$ and $Q_i$ introduced in~\eqref{d:yi}.
Therefore, Problem~\ref{prob:distributedOpt} can be recast as follows.
\begin{problem}\label{prob:ctr}
  Find a (dynamic) controller of the form
\begin{equation}\label{ctr:s:z}
\begin{array}{lcl}
  z^+ &=& \Phi z + B_x x + B_y y
  \\
  u &=& K_z z + K_x x + K_y y,
\end{array}
\end{equation}
	with state $z\in\R^{N\dz}$, $\dz\in\N$,  
	and a non-empty set of initial conditions $\cO\subset\R^{N\dx}\x\R^{N\dz}$
	such that, for each $\theta_0\in\R^p$, all the trajectories of the ``closed-loop system'' 
	\begin{equation}\label{s:cl_0}
	\begin{array}{lcl}
	x^+ &=& K_x x + K_z z + K_y y
	\\
	z^+ &=& \Phi z + B_x x + B_y y,
	\end{array}
	\end{equation}
	with $(x(0), z(0)) \in \cO$ are bounded and satisfy 
  \begin{equation*}
    \lim_{t\to \infty} x_i(t) = \ts = \Sigma \theta_0,
  \end{equation*}
  for each $i \in\until{N}$.\oprocend
\end{problem}

Restricting the regulator \eqref{ctr:s:z} to be linear is motivated by the fact
that, except for the affine term $Q \theta_0$ appearing in the output $y$, the
controlled system~\eqref{ctr:s:x} is \emph{linear}. Thus, Problem~\ref{prob:ctr}
results in a \emph{linear} set-point control problem that can be solved by a linear regulator.
In the same way, linearity implies that we can assume, without loss of generality, that the set of
initial conditions $\cO$ is an \emph{affine} subspace of $\R^{N(\dx+\dz)}$ whose
bias is parametrized by $\theta_0$, i.e., 
\begin{equation}\label{d:cO_V} 
\cO  :=   P\theta_0 + \cV     ,
\end{equation}
for some linear subspace $\cV\!$ of $\R^{N(\dx+\dz)}\!$ of dimension $n_v \!\in\! \N$, 
and for some matrix $P \! \in \! \R^{N(\dx+\dz) \x p}$ satisfying $\img P\subset \cV^\perp$. 

The closed-loop system \eqref{s:cl_0} can be compactly written as
\begin{equation}\label{ctr:s:cl}
\begin{bmatrix}
x\\z
\end{bmatrix}^+
=
F
\begin{bmatrix}
x\\z
\end{bmatrix}
+ G \theta_0,
\end{equation}
with 
\begin{equation*}
\begin{aligned}
F  &:=\begin{bmatrix}
K_x+K_y C & K_z\\ B_x+B_y C & \Phi 
\end{bmatrix}, 
& G &:= \begin{bmatrix}
K_y    \\ B_y    
\end{bmatrix}Q.
\end{aligned}
\end{equation*}
The gradient tracking algorithm~\eqref{eq:gradient_tracking_xz} exhibits the 
same closed-loop structure as~\eqref{ctr:s:cl}, in which the gradients act as an 
output feedback action.
In the following, we provide necessary and sufficient conditions for the
existence of a controller of the form \eqref{ctr:s:z} and a set $\cO$ of the
form \eqref{d:cO_V} solving Problem \ref{prob:ctr}.

\subsection{Necessary and Sufficient Conditions}
\label{sec:control_conditions}

Let $\nb:=N(\dx+\dz)$ and $n_v\le \nb$. Consider an $n_v$-dimensional vector 
subspace $\cV$ of $\R^\nb$
and let $T \in\R^{\nb\x \nb}$ be an orthonormal matrix of the form $T=\begin{bmatrix}
T_1 & T_2
\end{bmatrix}$, with   $T_1\in\R^{\nb \x n_v}$ and $T_2\in\R^{\nb \x (\nb-n_v)}$ 
satisfying
\begin{equation}\label{ctr:d:T1_T2}
  \begin{aligned}
    \img T_1 &= \cV, & \img T_2 &= \cV^\perp.
  \end{aligned}
\end{equation}
Then, it is easy to see that $\cV$ is $F$-invariant
if and only if %
\begin{equation*}
  F' := T^\top F T  
  = 
  \begin{bmatrix}
    F'_{I} & F'_{J}
    \\ 
    0 & F'_{E}
  \end{bmatrix},
\end{equation*}
for some $F_{I}'\in\R^{n_v\x n_v}$, $F_{J}'\in\R^{n_v\x (\nb-n_v)}$ and $F_{E}'\in\R^{(\nb-n_v)\x(\nb-n_v)}$. 
The matrices $F_I'$ and $F'_E$ represent the restriction of $F$ to $\cV$ and $\cV^\perp$, respectively. 
These matrices yield to the following definition.
\begin{definition}
  The subspace $\cV$ is said to be: 
  \begin{itemize}
  \item \emph{internally stable} if $F_{I}'$ is Schur; 
  \item \emph{externally anti-stable} if $F'_E$ has no eigenvalue 
  inside the open unitary disc.\oprocend
\end{itemize}
\end{definition}

The forthcoming proposition is the main result of the section and it states necessary and 
sufficient conditions for the existence of a controller of the form \eqref{ctr:s:z} and an initialization set of the form \eqref{d:cO_V}
solving Problem \ref{prob:ctr}.
For simplicity, although not necessary, we restrict the focus to initialization sets with 
the following additional property.
\begin{definition}\label{def:conditioned}
	A set $\cO$ of the form~\eqref{d:cO_V} is said to be an \conditionedinitializationset/
	if $\cV$ is $F$-invariant and externally anti-stable.\oprocend
\end{definition}
\begin{proposition} \label{prop:ctr}
  Consider a controller of the form~\eqref{ctr:s:z} resulting in the closed-loop system~\eqref{ctr:s:cl}.
  Let $\cO$ be an \conditionedinitializationset/ of the form~\eqref{d:cO_V}, for
  some $n_v$-dimensional $F$-invariant subspace $\cV$ of $\R^{\nb}$ and 
  some $P\in\R^{\nb \x p}$ satisfying $\img P\subset\cV^\perp$. Moreover, let $T_2 \in\R^{\nb \x (\nb-n_v)}$ 
  be an orthonormal matrix satisfying $\img T_2 = \cV^\perp$.
	Then, Problem~\ref{prob:ctr} is solved from $\cO$ by a controller of the 
	form~\eqref{ctr:s:z} if and only if
	\begin{enumerate}
		\item the set $\cV$ is internally stable;
		\item there exists  $\Pi\in\R^{\nb\x p}$ satisfying
		\begin{subequations}\label{e:regeq}
			\begin{align} 
			\Pi &=  F\Pi + G \label{regeq:Pi} \\[.2em]
			\bI \Sigma &=  \begin{bmatrix}
			I_{N\dx} & 0_{N\dx \x N\dz} 
			\end{bmatrix} \Pi \label{regeq:Psi}\\[.2em]
			0  &=    T_2^\top (\Pi-P).\label{regeq:T2} 
			\end{align}
		\end{subequations}
		\oprocend

	\end{enumerate}
\end{proposition}

Regarding the claim of Proposition~\ref{prop:ctr}, 
we observe that equation~\eqref{regeq:Pi} expresses the existence, 
for every $\theta_0\in\R^p$, of an equilibrium of the closed-loop system~\eqref{ctr:s:cl} 
given by
\begin{equation}\label{ctr:xzstar}
  (\xeq,\zeq) := \Pi\theta_0 .
\end{equation}
Equation \eqref{regeq:Psi}, instead, forces such equilibrium to be an optimal solution 
of problem~\eqref{eq:problem}, namely
  $\xeq= \bI \Sigma \theta_0 = \bI \ts$.
Finally, equation \eqref{regeq:T2} and the internal stability of $\cV$  express the fact that, if the closed-loop system \eqref{ctr:s:cl} is initialized with $(x(0),z(0)) \in \cO$,
then the equilibrium point~\eqref{ctr:xzstar} attracts all the closed-loop trajectories.

\section{Gradient Tracking Revisited}
\label{sec:GT_revisited}  
In this section, we establish a bridge between the gradient tracking distributed 
algorithm described in Section~\ref{sec:gradient_tracking} and the system theoretical framework 
discussed in Section~\ref{sec:control_approach}.
The design of a distributed optimization algorithm solving
	problem~\eqref{eq:problem} can be equivalently recast as the problem of finding
	a regulator of the form \eqref{ctr:s:z} which satisfies Problem~\ref{prob:ctr} 
	and is \emph{sparse}, in the sense that each control input $u_i$ depends
	\emph{only on the neighboring information} $(x_j,y_j)$, $j\in\nbrs_i$.
Specifically, we show that matrices $\Phi$, $B_x$, $B_y$, $K_z$, $K_x$, $K_y$ in
the controller~\eqref{ctr:s:z} 
can be properly chosen to implement a class of gradient tracking algorithms that, among others,
includes~\eqref{eq:gradient_tracking_xz}.
To this end, we progressively fix the available degrees 
of freedom in the controller~\eqref{ctr:s:z} with the aim of satisfying the conditions given in 
Proposition~\ref{prop:ctr}.

\subsection{Gradient Tracking as a Control System}

First we set the controller dimension equal to the plant dimension, 
i.e., $\dz = \dx$. Moreover, we let 
in the controller~\eqref{ctr:s:z}
\begin{align}\label{gt:choices}
  \Phi &= \btA, & B_x &=0, & B_y&= \btA-I, & K_x &:=\bA,
\end{align}
where $\bA \in\R^{N\dx\x N\dx}$ satisfies $\bA\bI = \bI$
and $\btA\in\R^{N\dx \x N\dx}$ $\bI^\top \btA= \bI^\top$, 
while $K_z$ and $K_y$ are still free.
We notice that all the matrices in~\eqref{gt:choices} are sparse, 
resulting in a controller that can be implemented in a 
fully distributed way.

The choice~\eqref{gt:choices} results in a closed-loop 
system~\eqref{ctr:s:cl} with 
\begin{align}\label{gt:d:FG}
\begin{aligned}
  F 
  &=
  \begin{bmatrix}
	  \bA + K_y C & K_z
	  \\
	  (\btA-I)C & \btA
  \end{bmatrix}, 
  & G &:= 
  \begin{bmatrix}
    K_y Q
    \\ 
    (\btA - I) Q
  \end{bmatrix}.
  \end{aligned}
\end{align}

In the following, we investigate conditions on the choice of $K_z$ and $K_y$ such 
that an \conditionedinitializationset/ $\cO$ and the controller~\eqref{ctr:s:z}
satisfy the assumptions of Proposition~\ref{prop:ctr}.
As a first result we claim the following.
\begin{lemma}\label{lemma:V}
	Consider the closed-loop system~\eqref{ctr:s:cl} in the setting described above.
	Then, 
  \begin{enumerate}
	\item there exists an  $(\nb-\dx)$-dimensional subspace $\cV$ of $\R^{\nb}$ 
	that is $F$-invariant and externally anti-stable for \emph{all the possible choices} of $K_y$ and $K_z$;
	\item there always exist $K_y$ and $K_z$ such that $\cV$ is also	internally stable.
  \end{enumerate}
\end{lemma}
\begin{proof}
We first notice that $F$ in~\eqref{gt:d:FG} can be decomposed 
in two terms as
\begin{equation}\label{gt:e:F_F0}
  F = 
  F_0  + B_0
  \begin{bmatrix}
    K_yC & K_z
  \end{bmatrix},
\end{equation}
where
\begin{equation*}
\begin{aligned}
F_0 &:= \begin{bmatrix}
\bA  & 0\\
(\btA - I) C  & \btA
\end{bmatrix}, 
& 
B_0 &:= 
\begin{bmatrix}
I\\ 0
\end{bmatrix}.
\end{aligned}
\end{equation*}
As a consequence $F$ can be thought of as being obtained by stabilizing the
following auxiliary system
\begin{equation*}
\xb_0^+ = F_0 \xb_0 + B_0\ub_0
\end{equation*}
by means of the state-feedback control law
$\ub_0 :=  
\begin{bmatrix}
K_yC & K_z
\end{bmatrix}
\xb_0$.
Being $F_0$ triangular, it holds that $\sigma(F_0)=\sigma(\bA)\cup\sigma(\btA)$. Hence,
$F_0$ has an eigenvalue equal to $1$ with algebraic multiplicity $2\dx$, while all 
the other eigenvalues lie inside the open unitary disc. It can be shown that a basis for the left-eigenspace of $F_0$ associated 
to the eigenvalue $1$ is given by the span of $v_1^\top$ and $v_2^\top$  defined as
\begin{equation}\label{gt:v1_v2}
\begin{aligned}
  v_1^\top &:= 
  \begin{bmatrix}
  v_{11}^\top&0
  \end{bmatrix}, & v_2^\top &:= 
  \begin{bmatrix}
    0 & \bI^\top  
  \end{bmatrix},
\end{aligned}
\end{equation}
where $v_{11}$ satisfies $v_{11}^\top \bA = v_{11}^\top$.
We further observe that the left-kernel of $\begin{bmatrix}
F_0 - I & B_0
\end{bmatrix}$ is spanned only by $v_2^\top$. Therefore, the stabilizability 
PBH test 
ensures that the non-reachable subspace of $(F_0,B_0)$ is a 
$\dx$-dimensional subspace of the eigenspace associated to the eigenvalue 
$1$  and, on the other hand, that the reachable subspace has dimension 
$n_v=\nb-\dx$. Therefore, point \emph{(i)} follows by taking $\cV$ equal to 
the reachable subspace, and by noticing that its $F$-invariance and external 
anti-stability properties cannot be changed via feedback, i.e., by any choice
for $K_z$ and $K_y$. 

To show point \emph{(ii)}, we resort to the reachability Kalman decomposition. 
Consider a transformation matrix $T := 
\begin{bmatrix}
T_1 & T_2
\end{bmatrix}$ with
\begin{align}\label{gt:T1_T2}
  \begin{aligned}
      T_1 & :=  
  \begin{bmatrix}
  I & 0 \\0 & R 
  \end{bmatrix}, 
    &   T_2 & := 
  \dfrac{v_2}{\sqrt{N}}
  =
  \dfrac{1}{\sqrt{N}}
    \begin{bmatrix}
    0 \\ \bI
  \end{bmatrix},
  \end{aligned}
\end{align}
where $R\in\R^{N\dx \x (N-1)\dx}$ is such that $RR^\top = I$ and $R^\top \bI=0$.
Then, it holds $T^{-1}=T^\top$ and $T$ transforms $(F_0,B_0)$ into 
$F_0':=T^\top F_0 T$ and $B_0':=T^\top B_0$ of the form
\begin{align*}
\begin{aligned}
F_0' & = 
\left[
\begin{array}{c c |c}
\bA & 0 & F_{J1} \\
-R^\top ( I - \btA)C & R^\top \btA R & F_{J2} \\
\hline
0 & 0 & I
\end{array}
\right]\!,
& \!\!
B_0' & = 
\left[
\begin{array}{c}
I\\0\\ \hline 0
\end{array}
\right ],
\end{aligned}
\end{align*}
for some $F_{J1}$ and $F_{J2}$.
Furthermore by construction the pair
\begin{align}  \label{d:Fi_Bi}
  (F_I',B_I'):=\left( 
  \begin{bmatrix}
  \bA & 0  
  \\
  -R^\top (I-\btA)C & R^\top \btA R   
  \end{bmatrix},\;  
  \begin{bmatrix}
  I\\0 
  \end{bmatrix}\right) 
\end{align}
is completely reachable, and being $C$ nonsingular, then there always exist gain matrices $K_y$ and $K_z$ 
satisfying 
$\begin{bmatrix}
K_y C & K_z
\end{bmatrix}T   =\begin{bmatrix}
K_{yI}' & K_{zI}' & K_J'
\end{bmatrix}$  
such that the matrix  $F_I' + B_I'\begin{bmatrix}
K_{yI}'  & K_{zI}'
\end{bmatrix}$ is Schur. Thus, point \emph{(ii)} follows since the latter condition 
implies that $\cV$ is internally stable.
\end{proof}

In the rest of the section, we denote by $\cV$ the subspace produced by
Lemma~\ref{lemma:V}.  The following result gives a sufficient condition on the
choice of $K_y$ and $K_z$ such that equations~\eqref{e:regeq} in
Proposition~\ref{prop:ctr} admit a solution.
\begin{lemma}\label{lemma:Pi}
	Consider the closed-loop system~\eqref{ctr:s:cl} in the setting described above.
  Pick $K_z$ and $K_y$ such that
	\begin{equation}\label{gt:Ky_eq_Kz}
	  K_z = K_y.
	\end{equation} 
	Then, there exist $\Pi\in\R^{\nb\x p}$ and $P\in\R^{\nb\x p}$, satisfying $\img P\subset\cV^\perp$, such that equations \eqref{e:regeq} hold.
\end{lemma}
\begin{proof}
Let $\Pi = \col(\Pi_x,\Pi_z)$, with 
$\Pi_x\in\R^{N\dx\x p}$ and $\Pi_z\in\R^{N\dx\x p}$. Then, 
in view of~\eqref{gt:d:FG}, $\Pi$ solves equations~\eqref{regeq:Pi} 
and~\eqref{regeq:Psi} if and only if
\begin{align}\label{gt:Pis_1}
\begin{array}{lcl}
  \Pi_x &=& (\bA+K_y C) \Pi_x + K_z \Pi_z + K_y Q 
  \\
  \Pi_z &=& (\btA-I)C \Pi_x + \btA\Pi_z + (\btA-I)Q 
  \\
  \bI \Sigma &=& \Pi_x .
\end{array}
\end{align}
By~\eqref{gt:Ky_eq_Kz} and since $(\bA-I) \bI =0$, 
we can rewrite~\eqref{gt:Pis_1} as
\begin{align}\label{gt:PIz} 
\begin{array}{lcl}
0 &=&   K_y (\Pi_z+C \bI \Sigma +Q)\\
0 &=& (\btA-I)( \Pi_z + C\bI \Sigma + Q ) \\
\Pi_x &=& \bI \Sigma.
\end{array}
\end{align}
Therefore, 
$\Pi_x=\bI \Sigma$ and $\Pi_z =-C\bI \Sigma -Q$ solve~\eqref{regeq:Pi}-\eqref{regeq:Psi}.

Finally, as for the existence of $P$ satisfying~\eqref{regeq:T2}, we observe that 
$T_2^\top$ is full rank and $\img T_2=\cV^\perp$. Hence, given $\Pi$, equation~\eqref{regeq:T2} is satisfied,
e.g., with $P=T_2 (T_2^\top T_2)^{-1} T_2^\top \Pi $ that fulfills $\img P\subset \img T_2=\cV^\perp$.
\end{proof}

\begin{remark}\label{rmk:Lemmas}
  While Lemma~\ref{lemma:V} is linked only to a \emph{stability} requirement on 
  the closed-loop system, the choice \eqref{gt:Ky_eq_Kz} of Lemma~\ref{lemma:Pi} 
  represents a constraint ensuring the \emph{existence of an equilibrium which is an 
  optimal solution for the optimization problem~\eqref{eq:problem}}. 
In fact, in order to obtain internal stability of $\cV$ we could, for example,
choose $K_z=0$ and $K_y$ any matrix so that $\bA+K_y C$ is Schur (which always
exists).  However, such a choice does not satisfy~\eqref{gt:Ky_eq_Kz} and,
hence, the resulting algorithm would not ensure the existence of an optimal
equilibrium for the closed-loop system. \oprocend
\end{remark}

In the following proposition we merge the previous results to give sufficient conditions
on the choice of $K_z$ and $K_y$ so that the trajectories of~\eqref{ctr:s:cl} initialized 
in $\cO = P\theta_0 + \cV$ converge to a solution of Problem~\ref{prob:ctr}.
\begin{proposition}\label{prop:gt}
	Consider the closed-loop system~\eqref{ctr:s:cl} in the setting described above.
	Let $K_z = K_y$ be such that $F$ in~\eqref{gt:d:FG}
	has all the eigenvalues but $\dx$ inside the open unitary disc.
	Then, Problem~\ref{prob:ctr} is solved from $\cO$,
  in the sense that all the trajectories of the closed-loop system \eqref{ctr:s:cl} 
  originating in $\cO$ are bounded and
   $ \lim_{t\to\infty} x_i(t) \!=\! \ts$, $\forall \, i=\toN$.\oprocend
\end{proposition}
\begin{proof}
  In view of Lemma~\ref{lemma:V}, there exists an $F$-invariant and externally
  anti-stable subspace $\cV$. Moreover, it is also internally stable whenever
  $K_z$ and $K_y$ are such that $F$ has all the eigenvalues but $\dx$ inside the
  open unitary disc.  In view of Lemma~\ref{lemma:Pi}, if $K_z = K_y$, then
  there exist $\Pi$ and $P$ such that steady-state condition~\eqref{e:regeq}
  hold.  Hence, the claim follows by Proposition~\ref{prop:ctr}.
\end{proof}

Finally, we notice that the choice of $K_z$ and $K_y$ in Proposition \ref{prop:gt} might not satisfy the network constraints. In the following, we discuss how the usual practice in distributed optimization
of selecting a common stepsize $\gamma>0$ for all the agents, is consistent 
with Proposition~\ref{prop:gt}, provided that 
$\gamma$ is taken sufficiently small. In our framework, this is achieved by setting $K_y = K_z = -\gamma I$.
In this way we complete the result of the section by showing that $K_z$ and $K_y$, 
fulfilling both the assumptions of Proposition \ref{prop:gt} and the network constraints, 
always exist.
The feasibility of this choices follows as a particular case of the following result.

\begin{proposition}\label{prop:gamma}
  Consider the closed-loop system~\eqref{ctr:s:cl} in the setting described
  above and let $K_z = K_y = -\Lambda$, with $\Lambda\in\R^{{N \dx} \x {N \dx}}$ diagonal and 
  positive definite.  Then, there exists $\gamma^\star>0$ such that, if 
  all the eigenvalues of $\Lambda$ lie in $(0,\gamma^\star)$, Problem~\ref{prob:ctr}
  is solved from $\cO$, in the sense that all the trajectories of the
  closed-loop system~\eqref{ctr:s:cl} originating in $\cO$ are bounded and
  $
    \lim_{t\to\infty} x_i(t) = \ts$, 
  $\forall \, i=\toN$.\oprocend
\end{proposition}

\subsection{Remarks About the Proposed Approach}
\label{sec:alg_remarks}
We start by pointing out some aspects related to the initialization.
We observe that, if $K_y$ is nonsingular, then the choice of $\Pi_z$
satisfying \eqref{gt:PIz} is unique and it is given by $\Pi_z = - (C \bI \Sigma+Q)$.
Thus, recalling the definition of $\Sigma$ in~\eqref{d:Sigma}, it holds
$\bI^\top  \Pi_z = - \bI^\top (C \bI \Sigma +Q)= 0$.
Moreover, we have shown that we can set the matrix $T_2^\top = v_2^\top = \begin{bmatrix}
  0 & \bI^\top 
\end{bmatrix}$.
Thus, equation~\eqref{regeq:T2} leads to 
$\begin{bmatrix}
  0 & \bI^\top 
\end{bmatrix} P = \bI^\top \Pi_z = 0$.
This means that the \conditionedinitializationset/ $\cO$ 
  coincides with $\cV$, 
  i.e., the distributed algorithm works \emph{only if} initialized in the 
  reachable subspace $\cV$, which means that $(x(0),z(0))$ has
  to be chosen so that
    $\sum_{i=1}^N z_i(0) = 0$.
  This, in turn, is consistent with Remark~\ref{rmk:gt_init}.
  However, we point out that $\sum_{i=1}^N z_i(0) = 0$ is necessary 
  only if $\ker K_y\cap\ker (\btA - I) = \{0\}$, as otherwise different
  choices of $\Pi_z$ might be possible.

  We underline that the only parameters of the problem
  \eqref{eq:problem}-\eqref{d:fi} that need to be known for the design of the
  gains $K_y$ and $K_z$ (fulfilling the stability requirement of Proposition
  \ref{prop:gt}) are the matrices $C_i$.
  Nevertheless, due to the continuity of the eigenvalues of the closed-loop
  matrix $F$ with respect to variations in $K_y$ and $K_z$, we also observe that
  whenever internal stability of $\cV$ is ensured for a ``nominal'' value
  $C_i^\circ$ of $C_i$, it also holds for all the actual values of $C_i$ in a
  sufficiently small open neighborhood of $C_i^\circ$.

  More in general, well-known results in the context of (hybrid) dynamical
  systems (see, e.g., \cite[Proposition 6.34]{Goebel2012}), show that any
  algorithm of the form \eqref{ctr:s:z} fulfilling the conditions of Proposition
  \ref{prob:ctr} is ``robust'' with respect to parameter perturbations and
  measurement noise. That is, for sufficiently small perturbations and noise,
  boundedness of the closed-loop trajectories is preserved, and the asymptotic
  error from the optimal solution is related to the noise bound.

  Finally, we underline how well-known arguments on homogeneous approximations
  of nonlinear systems (see, e.g., \cite[Theorem 9.11]{Goebel2012}) can be used
  to show that the global (in $\cO$) result of Proposition \ref{prob:ctr}
  implies a \emph{local} (in $\cO$) result when sufficiently regular
  nonlinearities comes into play. This, in turn, permits to extend the presented
  results ``locally'' to optimization problems of the form
  \eqref{eq:problem}-\eqref{d:fi} with nonlinear, strongly convex functions and 
  smooth $f_i$. 

\section{Conclusions}
\label{sec:conclusions}
In this paper we proposed a system theoretical approach to analyze a class of
gradient tracking algorithms for distributed quadratic optimization.
We formulated the design of a distributed algorithm %
as the design of a (linear) dynamic regulator solving a set-point
control problem.
We highlighted structural properties of the designed regulator
and we showed that they are fulfilled by the gradient tracking. 
Moreover, we proved how lack of reachability of the closed-loop
system imposes conditions on the initialization of distributed algorithms with
this structure.
The proposed system theoretical perspective suggests that robustness arguments,
customary in control theory, can be used to draw similar conclusions on the
optimization algorithms.
Finally, these results pave the way to more general technical tools for the
analysis of nonlinear, distributed optimization problems. \\[-3em]

\bibliographystyle{IEEEtran}
\bibliography{gradient_tracking}

\end{document}